%% file: elsarticle-template.tex
\documentclass[review]{elsarticle}

\makeatletter
\def\ps@pprintTitle{%
 \let\@oddhead\@empty
 \let\@evenhead\@empty
 \def\@oddfoot{\centerline{\thepage}}%
 \let\@evenfoot\@oddfoot}
\makeatother

\usepackage{ amsmath}
\usepackage{amssymb}
\usepackage{graphicx}
\graphicspath{ {images/} }
\usepackage{epstopdf}
\usepackage{color}

\usepackage{enumitem}
\usepackage{lineno,hyperref}

\newtheorem{lemma}{Lemma}
\newtheorem{theorem}{Theorem}

\newproof{proof}{Proof}
\newdefinition{remark}{Remark}
\newdefinition{example}{Example}

\begin{document}

\begin{frontmatter}

\title{Euclidean 1-center of a set of static and mobile points\tnoteref{mytitlenote}}

\tnotetext[fn1]{The first author would like to thank NBHM, DAE, Govt. of India and the second author would like to thank CSIR, Govt. of India for their financial support to carry out research. }

\author[math]{Kaustav Bose\corref{cor1}}
\ead{kaustavb@research.jdvu.ac.in}

\author[math]{Ranendu Adhikary}
\ead{ranendua@research.jdvu.ac.in}

\author[it]{Sruti Gan Chaudhuri}
\ead{srutiganc@it.jusl.ac.in}

\author[math]{Buddhadeb Sau}
\ead{bsau@math.jdvu.ac.in}

\address[math]{Department of Mathematics, Jadavpur University, West Bengal, India}

\address[it]{Department of Information Technology, Jadavpur University, West Bengal, India}

\cortext[cor1]{Corresponding author}


%
%

\begin{abstract}

In this paper, we consider the problem of computing the algebraic parametric equation of the Euclidean 1-center function in $\mathbb{R}^d$, $d \geq 2$, for a system of $n$ static points and $m$ mobile points having motion defined by rational parametric functions. We have shown that the corresponding Euclidean 1-center function is a piecewise differentiable function and have derived its exact parametric algebraic equation. If the positions of the static points and the rational parametric equations of the motion of the mobile points are given, we have proposed algorithms that compute the parametric equation of the Euclidean 1-center function.

\end{abstract}

\begin{keyword}
Euclidean 1-center \sep Smallest enclosing ball \sep Mobile Facility Location \sep Farthest-point Voronoi diagram
\end{keyword}

\end{frontmatter}

\section{Introduction}\label{intro}

The $\emph{k-center problem}$ is a fundamental combinatorial optimization problem in facility location. In the most general form, for a finite set of points $P$ in a metric space $(X, \delta)$ and a positive integer $k$, the $\emph{k-center}$ of $P$ is defined as a set $F \subseteq X$ of $k$ points, that minimizes $\max\limits_{p \in P} \min\limits_{c \in F} \delta(p, c)$. When the metric space is Euclidean it is called the Euclidean $k$-center. In the special case of $k = 1$, there is an explicit geometric characterization: the Euclidean $1$-center of a finite set $P$ is the center of the unique ball of smallest radius enclosing it. Thus computation of the Euclidean 1-center of a set of points is geometrically equivalent to finding its smallest enclosing ball. While there is a fairly rich literature devoted to the smallest enclosing ball problem for a set of static points, motivated by recent advances in mobile computing, telecommunication, and geographic information system, the problem has also attracted considerable interest for the mobile case.

Finding the Euclidean 1-center is fundamental in facility location problems. Facility location is a branch of operations research and computational geometry concerned with the optimal placement of one or more facilities in a way that minimizes the distance between the facilities and the clients. It encompasses a wide range of real life problems, including placement of manufacturing plants, warehouses, fire stations, hospitals, cell phone towers etc, to name a few. The problems of static facility location have been studied extensively. Recently, the classical problems of facility location have been posed in mobile setting \cite{BBKS00, BBKS06, Agarwal01, Agarwal05, Gao01, Durocher06}. Consequently, proximity queries like Euclidean $k$-center, Euclidean $k$-median etc for a set of moving points have attracted a lot of interest both from theoretical and applied perspectives. Apart from the facility location problems, the smallest enclosing ball, as a fundamental primitive in computational geometry, also has applications in various fields like computer graphics, robotics, military operations, data mining, and machine learning. For instance consider the mobile version of the ``bomb problem'': for a set of continuously moving targets, at any instant, the center of their smallest enclosing ball is the optimal place to drop a bomb of minimum strength in order to inflict maximum damage.

The earliest instance of the Euclidean $1$-center problem dates back to 1857 when the problem of finding the smallest enclosing disk of a set of $n$ points in the Euclidean plane was posed by Sylvester \cite{Sylvester}. For a long time, the algorithms developed were superlinear in $n$, the number of input points. A breakthrough was made by Megiddo \cite{Megiddo83} in 1983, when he first gave a deterministic optimal $O(n)$ algorithm for any fixed dimension $d$. However, Megiddo's algorithm was slow due to a $2^{2^d}$ dependence on the dimension $d$. Welzl \cite{Welzl91} proposed a simple randomized algorithm that runs in expected $O(n)$ time, based on the Linear Programming algorithm of Seidel \cite{Seidel90}. Fischer et al. \cite{Fischer03} presented a simple combinatorial algorithm, which based on a novel dynamic data-structure that can provide a fast and robust floating-point implementation, efficiently computing the smallest enclosing ball of point sets in dimensions up to 2,000.

\subsection{Related works}
The mobile version of the Euclidean 1-center problem was first introduced by Bespamyatnikh et al. \cite{BBKS00}. With mobile facilities, modeled by points having continuous motion, with their complete trajectory not fully known in advance (i.e., each mobile point follows a posted \emph{flight plan}, but can change its trajectory by submitting a \emph{flight plan update}), they proposed algorithms for maintenance of the Euclidean 1-center in the plane based on the\emph{ kinetic data structures} (KDS) introduced by Basch et al \cite{Basch99}. They showed that the mobile Euclidean 1-center may have unbounded velocity, i.e., for any $v \geq 0$ there is a set of three points in $\mathbb{R}^d$ ($d \geq 2$) with unit velocity that induces an instantaneous velocity greater than $v$ of the Euclidean 1-center. This result motivated the problem of finding bounded-velocity approximations of the mobile Euclidean 1-centre by other center functions like Steiner center, center of mass etc \cite{Durocher06}. Demaine et al. \cite{Demaine10} constructed a kinetic data structure for calculating the smallest enclosing ball (disk) of a set of moving points in the plane, based on certain properties of the farthest-point Delaunay triangulation initially suggested in \cite{Guibas98}. Their data structure generates $O(n^{3 + \epsilon})$ events for $n$ points with polynomial motion of fixed degree and has efficiency $O(n^{1 + \epsilon})$, which is the best that can be achieved for any data structure
based on farthest-point Delaunay triangulations. Constructing an efficient KDS for maintaining the smallest enclosing ball in higher dimensions ($d \geq 3$) is still an open issue \cite{Demaine10, Guibas04}. Recently Banik et al. \cite{Banik14} gave a complete geometric characterization of the locus of the mobile Euclidean 1-center for a set of $n$ static points $S$ and a single mobile point moving along a straight line $l$ in the Euclidean plane. They showed that the locus is continuous and piecewise differentiable linear and each of its differentiable pieces lies either on the edges of the farthest-point
Voronoi diagram of $S$, or on a line segment parallel to the line $l$. Given the positions of the static points, the locus of the mobile point (the straight line $l$), and the farthest-point Voronoi diagram of $S$, they proposed an $O(n)$ algorithm to compute the locus of the Euclidean 1-center.

\subsection{Our contribution}

To the best of our knowledge, most of the work hitherto done has been mainly directed towards effective maintenance of the Euclidean 1-center using kinetic data structures and constructing an efficient algorithm for higher dimensions is still an open issue. Banik et al. \cite{Banik14} gave an algorithm to find the path traced by the mobile Euclidean 1-center of a set of $n$ static points and a single mobile point moving along a straight line $l$ in the Euclidean plane. In this paper we have considered the problem of finding the algebraic parametric equation of the Euclidean 1-center function in the general setting, namely, for a system of $n$ static points and $m$ mobile points whose motions are defined by rational parametric curves in the $d$-dimensional Euclidean space. A parametric description in terms of time represents the motion of the Euclidean 1-center unambiguously and is decisive in a number of practical queries like its position at a time instant, velocity, etc. We have shown that the Euclidean 1-center function is a piecewise differentiable function and have derived its exact algebraic equation. If the positions of the static points and the rational parametric equations of the motion of the mobile points are given, we have proposed algorithms to compute the algebraic parametric equation of the Euclidean 1-center function.

First, in section \ref{defn} we give some preliminary definitions. In section \ref{theory} we prove certain properties of the Euclidean 1-center function. In particular we give a complete algebraic description of the function. Based on the theoretical results derived in section \ref{theory}, we discuss methods to compute the Euclidean 1-center function in section \ref{algo}.


\section{Definitions and notations}\label{defn}

 In this section we give some preliminary definitions and notations that will be frequently used throughout the paper. Let $P = \{p_1,p_2,\ldots,p_k\}$ be a set of points in $\mathbb{R}^d$, $d \geq 2$. A point $x \in \mathbb{R}^d$ is said to be an affine combination
of the points in $P$ if $x = \sum_{i=1}^{k}a_ip_i$, for some $a_i \in \mathbb{R}$ with $\sum_{i=1}^{k}a_i = 1$. The points in $P$ are called
\emph{affinely independent} if no point in $P$ can be written as an affine combination of the remaining points. In $\mathbb{R}^d$, at most $d + 1$ points
can be affinely independent. For $0 \leq k \leq d$, the set of all affine combinations of $k + 1$ affinely independent points in $\mathbb{R}^d$ is called a \emph{$k$-dimensional affine subspace} or a \emph{$k$-flat} in $\mathbb{R}^d$. In particular, a $0$-flat is a point, a $1$-flat is a straight line, a $2$-flat is a plane, etc. Two $k$-flats $A$ and $A'$ are said to be \emph{parallel}, if there is a vector $v$ such that $A' = A + v$. For a set $S \subseteq \mathbb{R}^d$, the \emph{affine hull} of $S$, denoted by $Aff(S)$, is the smallest affine subspace in $\mathbb{R}^d$ containing $S$ or equivalently, the set of all affine combinations of elements of $S$.

The \emph{convex hull} of a set $S \subseteq \mathbb{R}^d$, denoted by $Conv(S)$, is the smallest convex set containing $S$. The convex hull of a set of $k + 1$ affinely independent points is called a \emph{$k$-simplex}. Hence, a $0$-simplex is a point, a $1$-simplex is a line segment, a $2$-simplex is a triangle, a $3$-simplex is a tetrahedron, etc.  The convex hull of a set of $n+1$ ($n \leq k$) vertices of a $k$-simplex $\tau$ is called an \emph{$n$-face} of $\tau$.

A \emph{$d$-ball} in $\mathbb{R}^d$ of radius $r$ centered at $c$ is defined as $B_r(c) = \{x \in \mathbb{R}^d \mid \|x - c\| \leq r\}$. For $1 \leq k \leq d$, a \emph{$k$-ball} in $\mathbb{R}^d$ is defined as a proper intersection between a $k$-flat and a $d$-ball. By proper intersection we exclude the cases where the $k$-flat is tangent to the $d$-ball or the intersection is empty. Hence, a $1$-ball is a line segment, a $2$-ball is a disk, a $3$-ball is a solid sphere, etc.

For a set of points $P \subset \mathbb{R}^d$, the smallest enclosing ball of $P$, denoted by $SEB(P)$, is the smallest $d$-ball $B$ such that $P \subset B$. The smallest enclosing ball of a set of points exists uniquely. For a finite set of points $P \subset \mathbb{R}^d$, the circumball of $P$, denoted by $CB(P)$, is the smallest $d$-ball $B$ such that $P \subset \partial(B)$ and the center of $CB(P)$ is called the circumcenter of $P$, denoted by $cc(P)$. ($\partial(B)$ denotes the boundary of $B$.) The circumball of a set of points may not exist. However if the points in $P$ are affinely independent, then there exists a unique circumball.

Let $P = \{p_1,p_2,\ldots,p_n\}$ be a set of distinct points in $\mathbb{R}^d$, where $n \geq 2$. For $p_i \in P$, the \emph{farthest-point Voronoi cell} of $p_i \in P$ is the set of all points in $\mathbb{R}^d$ which are farther from $p_i$ than any other point in $P$, that is, $$\mathcal{FPV}or(p_i) = \{x \in \mathbb{R}^d \mid \|x - p_i\| \geq \|x - p_j\|, \forall i \neq j\}.$$ The \emph{farthest-point Voronoi diagram} generated by $P$, denoted as $\mathcal{FPVD}(P)$, is the partition of $\mathbb{R}^d$ into the farthest-point Voronoi cells of point of $P$. Points shared by two or more farthest-point Voronoi cells constitute the \emph{farthest-point Voronoi faces}. A set $P$ of distinct points in $\mathbb{R}^d$ will be said to be in \emph{general position} if no affinely dependent subset of $P$ lie on the boundary of any $d$-ball. If the points in $P$ are in general position, then a \emph{farthest point Voronoi $k$-face} of $\mathcal{FPVD}(P)$ is equidistant from exactly $d - k + 1$ points of $P$. So, a farthest point Voronoi $k$-face of $\mathcal{FPVD}(P)$ defined by a set of $d - k + 1$ points $P' \subseteq P$, denoted by $\mathcal{F}(P')$, is the set of all points $x \in \mathbb{R}^d$ such that $\|x - p_i\| = \|x - p_j\|, \forall p_i, p_j \in P'$ and $\|x - p_i\| > \|x - p_l\|, \forall  p_i \in P', p_l \in P \setminus P'$.



In this paper, we consider the problem of computing the algebraic parametric equation of the mobile Euclidean 1-center in $\mathbb{R}^d$, $d \geq 2$, for a system of $n$ static points and $m$ mobile points whose motion is defined by rational parametric curves. Here, by a rational parametric curve we mean a parametric curve $\nu = (\nu ^1,\dots,\nu^d) :I \rightarrow \mathbb{R}^d$  where each of its components is of the form  $\nu^i(t) = \frac{p^i(t)}{q^i(t)}$, with $p^i(t),q^i(t)$ being polynomials in $t$ with $q^i(t) \neq 0$. We consider $S = \{s_1,\ldots,s_n\} \subseteq \mathbb{R}^d$, the set of static points in general position and $V = \{\nu_1,\ldots,\nu_m\}$, the ``set'' of mobile points with each of its motion defined by rational parametric curves $\nu_k:I \rightarrow \mathbb{R}^d$, where $I$ is a compact (closed and bounded) interval. For a mobile point $\nu_i$, $\nu_i(t) \in \mathbb{R}^d$ is its position at the time instant $t \in I$.  By $V(t)$ we shall denote $\{\nu_1(t),\ldots,\nu_m(t)\}$, the set of positions of the mobile points at $t \in I$. Then the Euclidean 1-center function $\varepsilon:I \rightarrow \mathbb{R}^d$ is defined as $\varepsilon(t) =$ the center of $SEB(S \cup V(t))$.

\section{Characteristics of the Euclidean 1-center function}\label{theory}

%
\begin{lemma}\label{WhereItLies1}
For a system of $n$ static points $S$ and $m$ mobile points $V$, the Euclidean 1-center of $S \cup V(t)$ lies on a farthest point Voronoi face of $\mathcal{FPVD}(S)$ whenever there are at least two static points on the boundary of the smallest enclosing ball of $S \cup V(t)$.
\end{lemma}

\begin{proof}
Since the points of $S$ are in general position, at any time instant $t \in I$ at most $d + 1$ points of $S$ can lie on the boundary of the $SEB(S \cup V(t))$. Suppose that at some $t \in I$, precisely $k$ points of $S$, say $s_1, s_2, \ldots, s_k$, lie on  $\partial(SEB(S \cup V(t)))$, where $1 < k \leq d + 1$. Since $s_1, s_2, \ldots, s_k$ lie on the boundary of the smallest enclosing ball,
   $$\|\varepsilon(t) - s_1\| = \|\varepsilon(t) - s_2\| = \ldots = \|\varepsilon(t) - s_k\|,$$
 and as all other points of $S$ lie in the interior of the the smallest enclosing ball,
   $$\|\varepsilon(t) - s_i\| > \|\varepsilon(t) - s_j\|, \forall s_i \in \{s_1, s_2, \ldots, s_k \} \ \ and \ \ s_j \in S \setminus \{s_1, s_2, \ldots, s_k \}.$$
 
 Hence $\varepsilon(t)$ lies on the farthest point Voronoi face defined by $s_1, s_2, \ldots, s_k$ at $t$. $\Box$
\end{proof}

However, the converse of lemma \ref{WhereItLies1} is not true. The Euclidean 1-center $\varepsilon(t)$ may lie on a farthest point Voronoi face of $\mathcal{FPVD}(S)$, yet its boundary may not contain any static point at all. However, in case of a single mobile point we have the following `if and only if' condition. This will be useful during the computation of the Euclidean 1-center function in section \ref{algo}.

\begin{lemma} \label{WhereItLies2}
For a set of static points $S \subset \mathbb{R}^d$, and a single mobile point whose motion is given by $\nu :I \rightarrow \mathbb{R}^d$, the Euclidean 1-center of $S \cup \{\nu(t)\}$, at some instant $t \in I$, lies on a farthest point Voronoi face $\mathcal{F}(S')$, $S' \subseteq S$, if and only if $S' = S \cap \partial(SEB(S \cup \{\nu(t)\}))$, $|S'| \geq 2$.
\end{lemma}
\begin{proof}
Suppose that at some $t \in I$, $\varepsilon(t) \in \mathcal{F}(S')$, $S' = \{s_1, s_2, \ldots, s_k \} \subseteq S$, $1 < k \leq d + 1$. Since the boundary of the smallest enclosing ball of a set of points must contain at least two points, $S \cap \partial(SEB(S \cup \{\nu(t)\})) \neq \emptyset$. So let $s \in S \cap \partial(SEB(S \cup \{\nu(t)\}))$. If $s \notin S'$, then
$$\|\varepsilon(t) - s_i\| > \|\varepsilon(t) - s\|, \forall s_i \in S' = \{s_1, s_2, \ldots, s_k \},$$
since $\varepsilon(t) \in \mathcal{F}(S')$. But this contradicts the fact that $S' \subset SEB(S \cup \{\nu(t)\})$. Hence, $s \in S'$, say $s = s_1$. Now since we also have
$$\|\varepsilon(t) - s_1\| = \|\varepsilon(t) - s_2\| = \ldots = \|\varepsilon(t) - s_k\|,$$
$S' \subset \partial(SEB(S \cup \{\nu(t)\}))$. No other point of $S$ lies on $\partial(SEB(S \cup \{\nu(t)\}))$ as
$$\|\varepsilon(t) - s_i\| > \|\varepsilon(t) - s_j\|, \forall s_i \in \{s_1, s_2, \ldots, s_k \} \ \ and \ \ s_j \in S \setminus \{s_1, s_2, \ldots, s_k \}.$$
Hence, we have $S' = S \cap \partial(SEB(S \cup \{\nu(t)\}))$. The converse part immediately follows from the proof of lemma \ref{WhereItLies2}. $\Box$

\end{proof}

Now consider three points $A, B, C$ in the Euclidean plane. Their smallest enclosing ball (disk) has all of them lying on its boundary if and only if $ABC$ is an acute-angled or a right-angled triangle. If $ABC$ is an obtuse-angled triangle with $\angle BAC > \frac{\pi}{2}$, then only $B$ and $C$ lie on the boundary of the $SEB$, with $BC$ being the diameter. Hence in both cases it is easy to see that the the smallest enclosing ball is the circumball of the points on its boundary. This in fact holds true for any set of points even in higher dimensions. With constructions and arguments similar to that in \cite[p. 88]{Berg08} it can be easily proved that the smallest enclosing ball of a finite set of points is the smallest enclosing ball, and hence the circumball, of the points lying on its boundary.

\begin{lemma} \label{SEBandCB}
Let $S \subset \mathbb{R}^d$ be a finite set of points with $|S| \geq 2$ and $B_0 = SEB(S)$. If $T = S \cap \partial(B_0)$, then $B_0 = SEB(T) = CB(T)$.\end{lemma}

\begin{proof}

  We shall prove by contradiction. So let $B_0 \neq SEB(T)$. Let $B_1 = SEB(T)$ and hence is strictly smaller than $B_0$. Let $c_0$ and $c_1$ be the centers of $B_0$ and $B_1$ respectively. Clearly $T \subset B_0 \cap B_1$. Moreover, there is some $z \in T$ such that $z \in \partial(B_0) \cap \partial(B_1)$. We define a continuous family $\{B_\lambda \mid 0 \leq \lambda \leq 1 \}$ of $d$-balls as: $c_\lambda= (1-\lambda)c_0 + \lambda c_1$ is the center and $r_\lambda = \|c_\lambda - z\|$ is the radius of $B_\lambda$. Note that $B_0 \cap B_1 \subset B_\lambda$ $\forall$ $\lambda \in [0,1]$ and $B_\lambda$ is strictly smaller than $B_0$ $\forall$ $\lambda \in (0,1]$. Since $S \setminus T$ is contained in the interior of $B_0$, there is an $\epsilon \in (0,1)$ such that $S \setminus T \subset B_\epsilon$. Also $T \subset B_0 \cap B_1 \subset B_\epsilon$. Therefore $S \subset B_\epsilon$. But $B_\epsilon$ is strictly smaller than $B_0$ contradicting the fact that $B_0 = SEB(S)$.  $\Box$

\end{proof}

\begin{figure}[htb]
\centering
\def\svgwidth{0.8\textwidth}
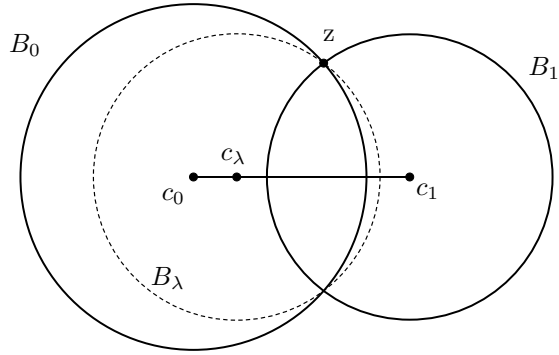
\caption{The continuous deformation of balls as described in lemma \ref{SEBandCB}}
\end{figure}

\begin{lemma}\label{maxaff}
Let $P$ be a set of co-spherical points in $\mathbb{R}^d$. If $P'$ is a maximum affinely independent subset of $P$, then $cc(P) = cc(P')$.
\end{lemma}

\begin{proof}

Let $dim(Aff(P)) = k$. If the points of $P$ are seen as points of $\mathbb{R}^k \cong Aff(P)$, then the circumball of $P$ in $\mathbb{R}^k$ is the $k$-ball $B_k = B_d \cap  Aff(P)$, where $B_d$ is the circumball of $P$ in $\mathbb{R}^d$ with $B_d$ and $B_k$ having coinciding center and equal radius. So without loss of generality we can assume that $dim(Aff(P)) = d$. Then $|P'| = d + 1$. But there is only a unique $d$-ball whose boundary passes through a set of $d + 1$ affinely independent points in $\mathbb{R}^d$. Hence $CB(P) = CB(P')$ and $cc(P) = cc(P')$. $\Box$


\end{proof}

  We shall exploit lemma \ref{SEBandCB} and \ref{maxaff} to derive an algebraic formulation of the Euclidean 1-center function. But first, we need to find an expression for the circumcenter of a simplex. Expressions for the circumcenter of a triangle in $\mathbb{R}^2$ or a tetrahedron  in $\mathbb{R}^3$ are well-known \cite{Shewchuk99}. The circumcenter of a $d$-simplex in $\mathbb{R}^d$ can be easily formulated in similar fashion. But a $k$-simplex in $\mathbb{R}^d$ has infinitely many $d$-balls whose boundary passes through its vertices when $k < d$. This is because there are infinitely many points in $\mathbb{R}^d$ that are equidistant from $k+1$ affinely independent points where $k < d$. However the circumcenter of a simplex is the unique point in its affine hull that is equidistant from each of its vertices \cite[p. 157-158]{Hales12}. Using the same approach as employed in \cite{Hales12} the circumcenter of a simplex can be computed by solving a system of linear equations even when the dimension of the ambient space is higher than the dimension of the simplex.

\begin{lemma}\label{cc}
  Let $\tau$ be a $k$-simplex in $\mathbb{R}^d$ with vertices $p_0, p_1, \ldots, p_k$, where $1 \leq k \leq d$. Then its circumcenter $c$ is given by
  \begin{equation}\label{cceqn}
    c = \sum_{j = 1}^{k}\frac{det(M_j)}{det(M)}(p_j - p_0) + p_0
  \end{equation}
   where
  \begin{equation}\label{matrix}
  M =
\begin{pmatrix}
(p_1 - p_0).(p_1 - p_0) &&  (p_1 - p_0).(p_2 - p_0) &  \ldots & (p_1 - p_0).(p_k - p_0)\\
(p_2 - p_0).(p_1 - p_0) &&  (p_2 - p_0).(p_2 - p_0) &  \ldots & (p_2 - p_0).(p_k - p_0)\\
\vdots && \vdots & \ddots & \vdots\\
(p_k - p_0).(p_1 - p_0) &&  (p_k - p_0).(p_2 - p_0) &  \ldots & (p_k - p_0).(p_k - p_0)
\end{pmatrix}
\end{equation}
and $M_j$  is the matrix formed by replacing the $j$th column of $M$ by the column vector
$\begin{pmatrix}\frac{(p_1 - p_0).(p_1 - p_0)}{2} \\ \frac{(p_2 - p_0).(p_2 - p_0)}{2} \\ \vdots \\ \frac{(p_k - p_0).(p_k - p_0)}{2} \end{pmatrix}$

  \end{lemma}

\begin{proof}
    The circumcenter of $\tau$, c, lies in the affine hull of $\{p_0, p_1, \ldots, p_k\}$. Translate the origin of the coordinate system to $p_0$ and write $p_i' = p_i - p_0$.
    Then the affine hull of $\{p_0, p_1, \ldots, p_k\}$ in the previous coordinate system is the linear subspace generated by $\{p_1', \ldots, p_k'\}$ in the new coordinate system.
    Let $c'$ be the circumcenter of $\tau$ in the new coordinate system, i.e., $c' = c - p_0$.
    Then $c'$ lies in the linear subspace generated by $\{p_1', \ldots, p_k'\}$. Since $\{p_0, p_1, \ldots, p_k\}$ is affinely independent, $\{p_1 - p_0, \ldots, p_k - p_0\} = \{p_1', \ldots, p_k'\}$ is linearly independent.
    Therefore, $c'$ can be uniquely written as $c' = \sum_{j = 1}^{k}\lambda_jp_j'$, for some $\lambda_1, \ldots, \lambda_k \in \mathbb{R}$.

    Since $c'$ is equidistant from $p_0' = \textbf{0}, p_1', \ldots, p_k'$, the vectors $p_i'$ and $c' - \frac{p_i'}{2}$ are orthogonal to each other for every $i = 1, 2, \ldots, k$. Hence for all $i = 1, 2, \ldots, k$ we have,
    $$ (c' - \frac{p_i'}{2}).p_i' = 0 $$
    $$\Rightarrow c'.p_i' = \frac{p_i'.p_i'}{2} $$
    $$\Rightarrow (\sum_{j = 1}^{k}\lambda_jp_j').p_i' = \frac{p_i'.p_i'}{2} $$

    Thus we obtain the following system of linear equations.
\begin{equation}\label{cc2}
\begin{pmatrix}
2p_1'.p_1' & 2p_1'.p_2' & \ldots & 2p_1'.p_k'\\
2p_2'.p_1' &  2p_2'.p_2' &  \ldots & 2p_2'.p_k'\\
\vdots & \vdots & \ddots & \vdots\\
2p_k'.p_1' &  2p_k'.p_2' &  \ldots & 2p_k'.p_k'
\end{pmatrix}
\begin{pmatrix}\lambda_1 \\ \lambda_2 \\ \vdots \\ \lambda_k \end{pmatrix} =
\begin{pmatrix}p_1'.p_1' \\ p_2'.p_2' \\ \vdots \\ p_k'.p_k' \end{pmatrix}
\end{equation}

The matrix
$$ M = \begin{pmatrix}
p_1'.p_1' &  p_1'.p_2' &  \ldots & p_1'.p_k'\\
p_2'.p_1' &  p_2'.p_2' &  \ldots & p_2'.p_k'\\
\vdots & \vdots & \ddots & \vdots\\
p_k'.p_1' &  p_k'.p_2' &  \ldots & p_k'.p_k'
\end{pmatrix} $$
can be written as $M = PP^T$, where $P$ is the matrix with row vectors $p_1', p_2', \ldots, p_k'$. As $p_1', p_2', \ldots, p_k'$ are linearly independent, it is easy to see that the rank of $M$ is $k$, and hence $det(M) \neq 0$. Thus the system of linear equations (\ref{cc2}) can be solved for $\lambda$ to get the desired expression. $\Box$

  \end{proof}

  For a set $\{p_0,p_1,\ldots,p_k\}$ of points in $\mathbb{R}^d$, we shall denote by
  $M(p_0,p_1,\ldots,p_k)$ the matrix in eqn. $(\ref{matrix})$ and by $\xi(p_0,p_1,\ldots,p_k)$ the expression in eqn. $(\ref{cceqn})$. It is easy to see that $M(p_0,p_1,\ldots,p_k)$ and $\xi(p_0,p_1,\ldots,p_k)$ are invariant under any permutation of $\{p_0,p_1,\ldots,p_k\}$.

\begin{theorem}\label{Main1}

  For a system of static points $S = \{s_1,\ldots,s_n\}$ in general position and mobile points $V = \{\nu_1,\ldots,\nu_m\}$ each having motion defined by  rational parametric curves $\nu_k:I \rightarrow \mathbb{R}^d$, the Euclidean 1-center function $\varepsilon:I \rightarrow \mathbb{R}^d$ is piecewise differentiable.
\end{theorem}

\begin{proof}

We decompose $I$ into closed intervals $I_1, I_2, \ldots$ with $I = \bigcup_k I_k$ and $|I_k \cap I_{k + 1}| = 1$ so that for all $t \in {I_k}^o$, $\partial(B(t))$ contains a fixed set of points of $S \cup V$. (${I}^o$ denotes the interior of $I$.) We call the end points of the intervals \emph{event points}.

\begin{description}
 \item[Case 1]

 Let $\nu_j(t) \notin \partial(B(t))$ for all $t \in {I_k}^o$ and $j=1, \ldots, m$. Then $\varepsilon(t)$ is a constant function, and hence differentiable in ${I_i}^o$.

\end{description}

\begin{description}
 \item [Case 2]

 Suppose that for $t \in {I_k}^o$, static points $s_1,\ldots, s_i$ and mobile points $\nu_1,\ldots,\nu_j$ are on the boundary of $B(t)$. Furthermore, assume that the points $s_1,\ldots, s_i,$ $\nu_1(t),\ldots,\nu_j(t)$ remain affinely independent for each $t \in {I_k}^o$. Then $det(M(s_1,\ldots, s_i,\nu_1(t),\ldots,\nu_j(t)))$ is non-vanishing in ${I_k}^o$. By lemma \ref{SEBandCB} for each $t \in {I_i}^o$, $\varepsilon(t)$ is the circumcenter of the simplex with vertices $\{s_1,\ldots, s_i,\nu_1(t),\ldots,\nu_j(t)\}$. So, $\varepsilon(t)= \xi(s_1,\ldots, s_i,\nu_1(t),\ldots,\nu_j(t))$ when $t \in {I_i}^o$, and hence differentiable since it is a rational function.

\end{description}

\begin{description}
 \item [Case 3]

Suppose that for $t \in {I_k}^o$, static points $s_1,\ldots, s_i$ and mobile points $\nu_1,\ldots,\nu_j$ are on the boundary of $B(t)$, but they do not remain affinely independent. But since $det(M(s_1,\ldots, s_i,\nu_1(t),\ldots,\nu_j(t)))$ is a rational function in $t$, either it is identically zero or it has only finitely many roots. If it has finitely many roots in ${I_k}^o$, then clearly the 1-center function is piecewise differentiable in ${I_k}^o$. If $det(M(s_1,\ldots, s_i,\nu_1(t),\ldots,\nu_j(t)))$ is identically zero, then by lemma \ref{maxaff} we have to take a maximum subset of $\{s_1,\ldots, s_i, \nu_1,\ldots,\nu_j\}$ that remain affinely independent and the situation reduces to the previous cases.
$\Box$
\end{description}

\end{proof}

The proof of theorem \ref{Main1} leads to the following result, which provides a complete characterization of the Euclidean 1-center function.

\begin{theorem}\label{Main2}
  For a system of static points $S = \{s_1,\ldots,s_n\}$ in general position and mobile points $V = \{\nu_1,\ldots,\nu_m\}$ each having motion defined by  rational parametric curves $\nu_k:I \rightarrow \mathbb{R}^d$, the 1-center function $\varepsilon:I \rightarrow \mathbb{R}^d$ is given by $\varepsilon(t) = \xi(P(t))$, where $P(t)$ is a maximum affinely independent subset of $(S \cup V(t)) \bigcap \partial(SEB(S \cup V(t)))$.
\end{theorem}

 \begin{remark}
In theorem \ref{Main1} we have shown that the center function $\varepsilon(t)$ is piecewise differentiable and have characterized the points, namely the event points, where $\varepsilon(t)$ may fail to be differentiable. The following are an instance where the center function is not differentiable at an event point, and another instance where it is differentiable at an event point. In fact, the center functions of both the examples have the same trace in $\mathbb{R}^2$, yet have different differentiability properties.


\begin{figure}[!htb]
\centering
\includegraphics[width=0.8\textwidth]{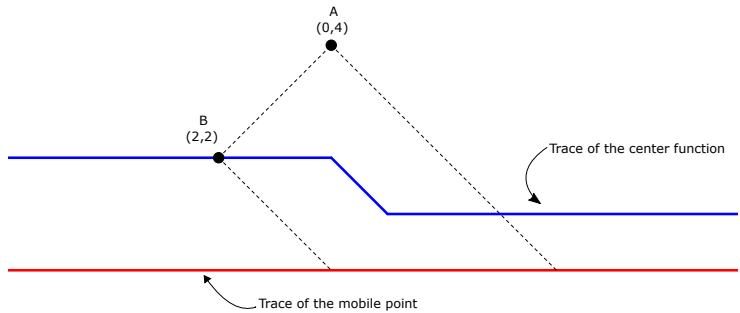}
\caption{Trace of the mobile point and the center function in example 1 and 2}
\label{Example03}
\end{figure}

\begin{example}\label{eg1}

We consider an example in $\mathbb{R}^2$ with a set of two static points and one mobile point. Let $A (0,4)$ and $B(-2,2)$ be the two static points and let the mobile point $\nu$ whose motion is given by the parametric curve $\nu: \mathbb{R} \rightarrow \mathbb{R}^2$ given by $\nu(t) = (t,0)$. Then the center function $\varepsilon(t)$ is given by
\begin{equation*}
  \varepsilon(t) = \left\{
                     \begin{array}{ll}
                       (\frac{t}{2},2), & \hbox{$t \in (-\infty,0]$;} \\
                       (\frac{t^2}{2(t+4)},\frac{-t^2+4t+16}{2(t+4)}), & \hbox{$t \in [0,4]$;} \\
                       (\frac{t-2}{2},1), & \hbox{$t \in [4,\infty)$.}
                     \end{array}
                   \right.
\end{equation*}

We check for differentiability at the event point $t = 0$.
$$\lim_{h \rightarrow 0, h>0} \frac{\varepsilon(h)-\varepsilon(0)}{h} = (0,0)$$
$$\lim_{h \rightarrow 0, h>0} \frac{\varepsilon(-h)-\varepsilon(0)}{-h} = (\frac{1}{2},0)$$
Hence $\varepsilon(t)$ is not differentiable at $t=0$.

\end{example}

\begin{example}\label{eg2}
We have the same set of static points as in Example \ref{eg1}. Here suppose that the motion of the mobile point is given by the parametric curve $\nu: \mathbb{R} \rightarrow \mathbb{R}^2$ given by $\nu(t) = (t^3,0)$. In this case the center function is
\begin{equation*}
  \varepsilon(t) = \left\{
                     \begin{array}{ll}
                       (\frac{t^3}{2},2), & \hbox{$t \in (-\infty,0]$;} \\
                       (\frac{t^6}{2(t^3+4)},\frac{-t^6+4t^3+16}{2(t^3+4)}), & \hbox{$t \in [0,2^{2/3}]$;} \\
                       (\frac{t^3-2}{2},1), & \hbox{$t \in [2^{2/3},\infty)$.}
                     \end{array}
                   \right.
\end{equation*}
Again $t = 0$ is an event point.
$$\lim_{h \rightarrow 0, h>0} \frac{\varepsilon(h)-\varepsilon(0)}{h} = (0,0)$$
$$\lim_{h \rightarrow 0, h>0} \frac{\varepsilon(-h)-\varepsilon(0)}{-h} = (0,0)$$
Hence in this case $\varepsilon(t)$ is differentiable at $t=0$.
\end{example}

\end{remark}

\section{Computation of the Euclidean 1-center function}\label{algo}


In theorem \ref{Main2} we have derived an exact algebraic description of the Euclidean 1-center function. The Euclidean 1-center at any time $t \in T$ is equal to the circumcenter of the points on the boundary of smallest enclosing ball at $t$. So in order to compute the Euclidean 1-center function, we need the information about the points that appear on the boundary of smallest enclosing ball at any instant throughout the parameter interval. The following result provides a necessary and sufficient condition for the input points to lie on the boundary of smallest enclosing ball.

\begin{lemma} \label{Alg1}

For $T \subseteq S \subset \mathbb{R}^d$, $CB(T) = SEB(S)$ if and only if $S \subset CB(T)$ and $cc(T) \in conv(T)$.
\end{lemma}

\newproof{proof5}{Proof}
\begin{proof}
  We have the following result from \cite{Fischer03}: for a set of points $T$ on the boundary of some ball $B$ with center $c$, $B = SEB(T)$ if and only if $c \in conv(T)$.

  Let $T \subseteq S$. If $SEB(S) = CB(T) = B$, then $B = SEB(T)$ by lemma \ref{SEBandCB}. Hence, $cc(T) \in conv(T)$. Conversely, let $S \subset CB(T)$ and $cc(T) \in conv(T)$. Hence $CB(T) = SEB(T)$. Since $T \subseteq S$ and $S \subset CB(T)$, we have $CB(T) = SEB(S)$. $\Box$
\end{proof}

Since the Euclidean 1-center at any time instant is determined by the points on the boundary of the $SEB$, a change in the algebraic description of the 1-center function is caused by some combinatorial change of the points on the boundary of the deforming $SEB$, i.e., some points appear on or leave the boundary of the $SEB$. Now if the location of the static points and the equation of the motion of the mobile points are known, then we have the equation of all the possible \emph{arcs} or curve pieces of the Euclidean 1-center function, using the formula in lemma \ref{cc}. Since the Euclidean 1-center function is continuous \cite{Durocher06}, it can go from one such curve to another only at an intersection point. Our algorithm is based on this basic idea. In short we shall iteratively compute the equation of the Euclidean 1-center function as: in each iteration given an arc of the 1-center function we compute its intersection with all the candidates for the subsequent arc and decide the next arc.

We shall first describe an algorithm in detail for the case where there is only a single mobile point and then modify it for the case of multiple mobile points.

%
%
%

\subsection{Outline of the algorithm in case of a single mobile point}\label{algo1}

Consider a set of $n$ static points $S \subset \mathbb{R}^d$ and a single mobile point $\nu$ whose equation of the motion is given by $\nu:[0,T] \rightarrow \mathbb{R}^d$. Suppose that at some instant, the equation of the 1-center function is given by $\xi(S', \nu(t))$, where $S' \subset S$. Then the events that can cause a change of the equation of the 1-center function at some $t_0 \in I$ can be characterized as the following:

\begin{enumerate}
  \item A set of static points $S'' \subset S$ appears on the boundary of the deforming $SEB$ at $t_0 \in I$. In that case the curve $\xi(S', \nu(t))$ intersects the farthest point Voronoi face defined by $S' \cup S''$ at $t_0$.

  \item A set of static points $S'' \subset S'$ leaves the boundary of the $SEB$ at $t_0 \in I$. Then there is an intersection between $\xi(S', \nu(t))$ and $\xi(S' \setminus S'', \nu(t))$ at $t_0$.

  \item The mobile point leaves (when it enters $SEB(S)$) or reappears (when it exits $SEB(S)$) on the boundary of the deforming $SEB$ at $t_0 \in I$.
\end{enumerate}

Throughout the remainder of the paper by an \emph{intersection point} or simply an \emph{intersection} between two parametric curves $\varphi_1:I \rightarrow \mathbb{R}^d$ and $\varphi_2:I \rightarrow \mathbb{R}^d$ we shall mean a parameter value $t_0$ such that $\varphi_1(t_0) = \varphi_2(t_0)$. Similarly an intersection point between a parametric curve $\varphi:I \rightarrow \mathbb{R}^d$ and some geometric object given in implicit form by $f(x_1,x_2,\ldots,x_d) = 0$ will refer to a parameter value $t_0$ such that $f(\varphi_1(t_0)) = 0$. The combinatorial changes of the points on the boundary of the deforming $SEB$ that can change the equation of the 1-center function can only occur at the intersection points characterized above. The parameter value $t$ corresponding to these events will be called an \emph{event point}. The algebraic curve between two such consecutive event points will be referred to as an \emph{arc} of the Euclidean 1-center function. The set of points on the boundary of the smallest enclosing ball that defines an arc will be referred to as the \emph{support set} of that arc.

Roughly speaking, the basic idea of our algorithm is that in each step with the present arc of the 1-center known, we find the next event point by computing the aforesaid intersections and determine the subsequent arc. The program terminates when no such event point is found in $[0, T]$. Note that since the mobile point is moving along an arbitrary curve, the behaviour of the 1-center at an event point would entirely depend upon the local properties of the curve, namely the direction of the mobile point at that point. Figure \ref{epsilon} illustrates an instance where small perturbations of the mobile point in different directions lead to different outcomes. To resolve this ambiguity, at these intersection points or event points we check which of the points remain on the boundary of the $SEB$ upon a small $\epsilon$ perturbation.

Finally we briefly sketch an outline of our proposed algorithm. Details regarding the implementation of certain steps of the algorithm is given in the next section. Suppose that we are given as inputs the set $S = \{s_1,s_2,\ldots,s_n\}$ of $n$ static points and the equation of the motion $\nu:[0,T] \rightarrow \mathbb{R}^d$ of the mobile point $\nu$. Assume that the farthest-point Voronoi diagram of $S$ is given or already computed. Software packages like Qhull \cite{Qhull} can compute farthest-point Voronoi diagrams in higher dimensions. Also assume that we are given the smallest enclosing ball of $S$ and $S \cup \{\nu(0)\}$. Hence for initialization, the first arc of the 1-center function is known from $SEB(S \cup \{\nu(0)\})$. So instead of the information of $SEB(S \cup \{\nu(0)\})$, we can simply assume that the equation of the first arc of the 1-center function and its support set is given.

\begin{description}
  \item[STEP 1.] Compute the sorted (in increasing order) list $L$ of intersections of $\nu(t)$ with $\partial(SEB(S))$. Label the intersections as `IN' or `OUT' depending on whether the mobile point enters or exits the $SEB(S)$ respectively. \\

  \item[STEP 2.] Insert $0$ into the empty list $E$ of `event points'. Label $0$ as `IN' if $\nu(0)$ is in the interior of $SEB(S)$. \\

  \item[STEP 3.] \textbf{If} the last `event point', say $t_{e_{last}}$, is an `IN'-point, the next event point is the next `OUT'-point in $L$.

  \textbf{Else}, the present arc of the center function is $\xi(S', \nu(t))$ for some $S' \subset S$. Then

  \begin{enumerate}[label=(\alph*)]
    \item for each nonempty $S'' \subset S'$, compute intersections of $\xi(S'', \nu(t))$ and $\xi(S', \nu(t))$.
    \item compute intersections of $\xi(S', \nu(t))$ and the farthest-point Voronoi faces involving $S'$, i.e., $\mathcal{FPVD}(P)$ with $S' \subset P$.
  \end{enumerate}


  The next `event point', say $t_{e_{new}}$, is one of these intersections or an `IN'-point in $L$ which ever is earliest in $(t_{e_{last}}, T]$.\\

  \item[STEP 4.] \textbf{If} no `event point' is found, \textbf{END} the program.

  \textbf{Else if} the `event point' $t_{e_{new}}$ is an `IN'-point, the next arc is the constant function whose value is the center of $SEB(S)$. Insert $t_{e_{new}}$ into $E$. \textbf{Go to} Step 3.

  \textbf{Else}, find the `support set' of $SEB(S \cup \{\nu(t_{e_{new}} + \epsilon)\})$, for a predefined sufficiently small $\epsilon > 0$, using lemma \ref{Alg1}. Compute the new arc using the formula in lemma \ref{cc}. Insert $t_{e_{new}}$ into $E$. \textbf{Go to} STEP 3.

\end{description}


\begin{figure}[!htb]
\centering
\includegraphics[width=0.45\textwidth]{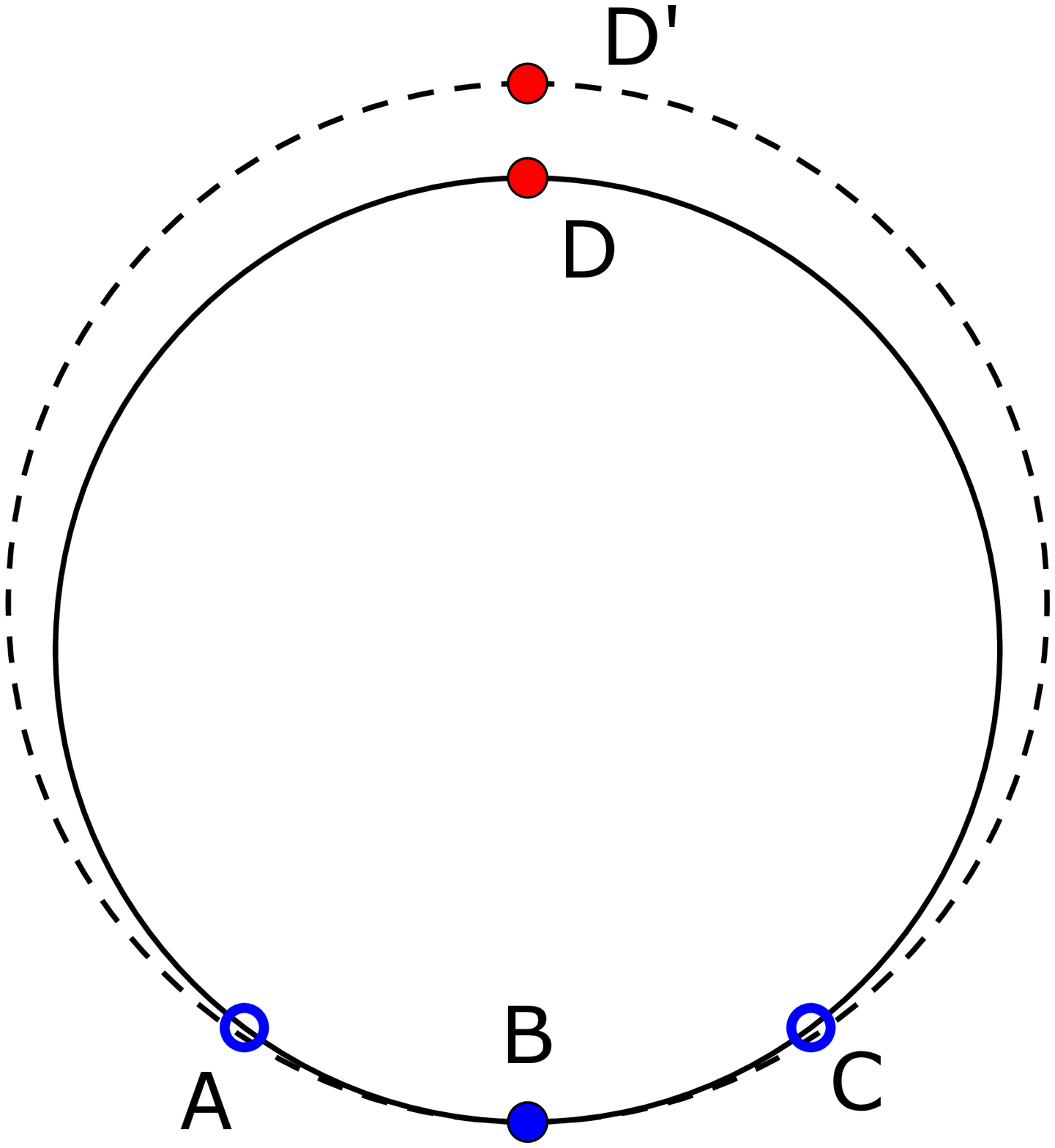}
\hfill
\includegraphics[width=0.45\textwidth]{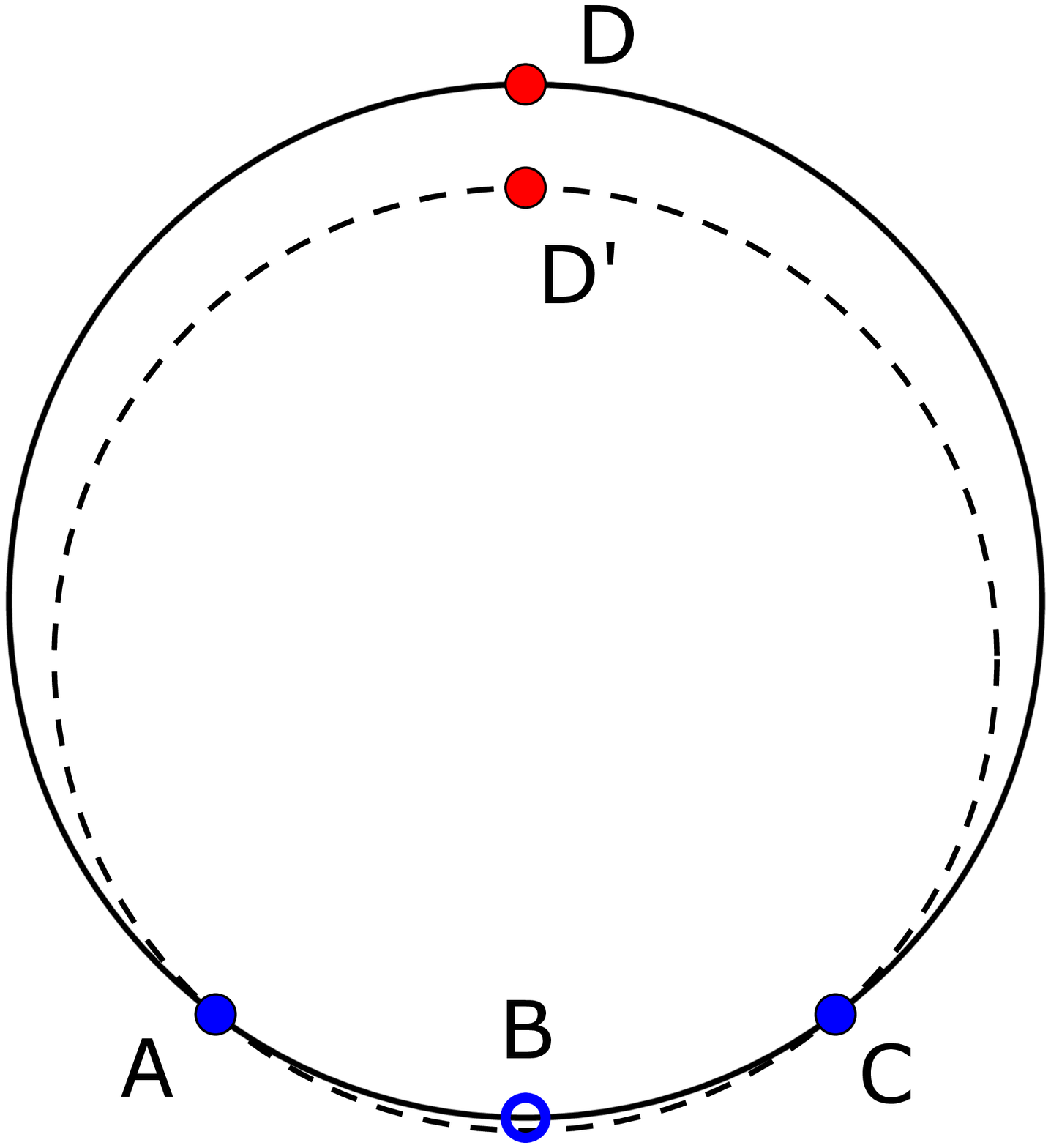}\\
\includegraphics[width=0.45\textwidth]{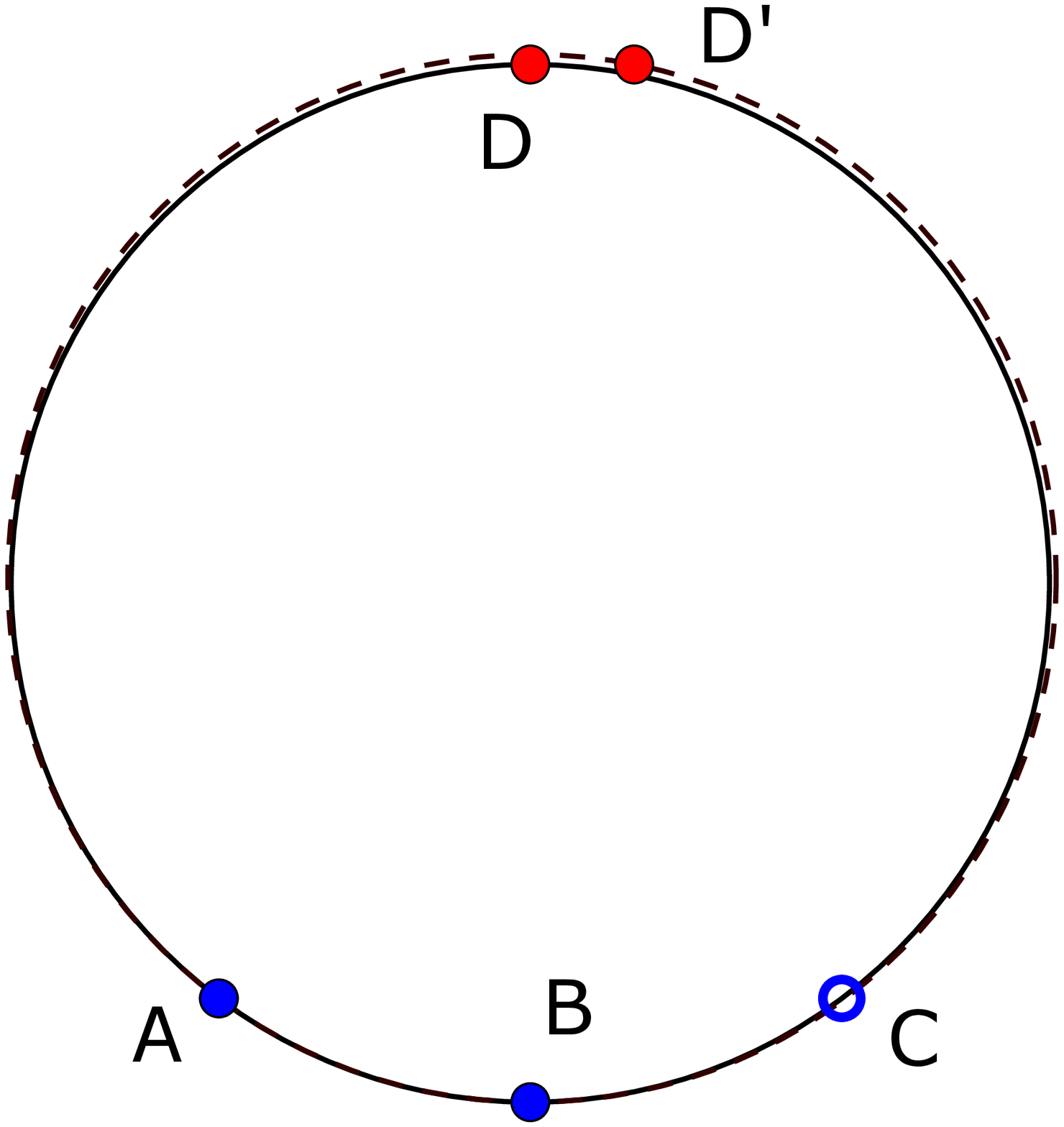}
\hfill
\includegraphics[width=0.45\textwidth]{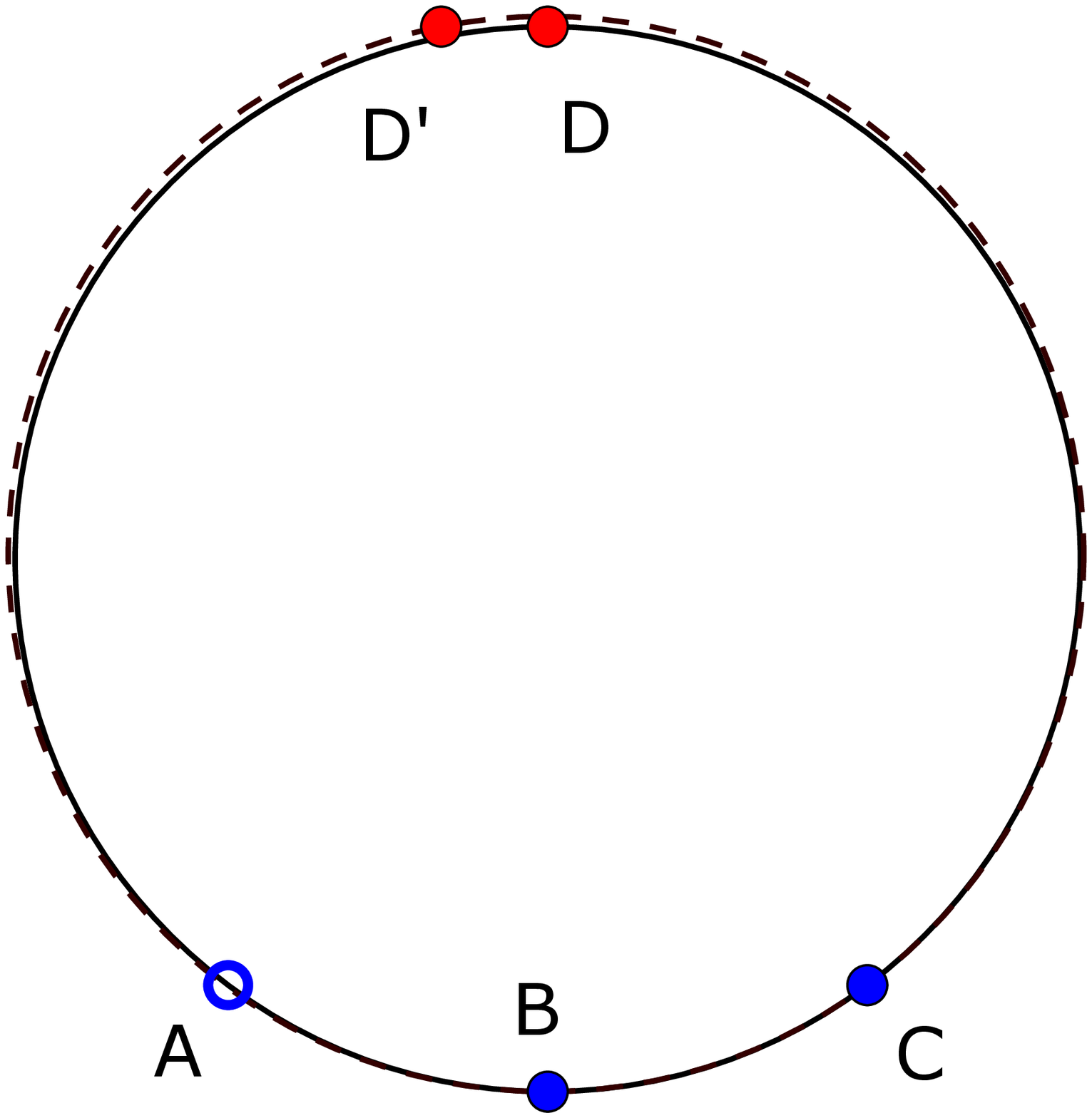}\\
\includegraphics[width=0.7\textwidth]{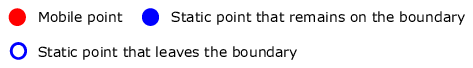}
\caption{Different outcomes of an infinitesimal perturbation of a mobile point}\label{epsilon}
\end{figure}

\subsection{Implementation and analysis}\label{analysis1}

Our proposed algorithm requires computation of intersections between two rational parametric curves and between a rational parametric curve and a farthest-point Voronoi face. Consider two rational parametric curves given by $\varphi_1(t) = \left(\frac{a_i(t)}{c_i(t)}\right)_{i=1}^{d}$ and $\varphi_2(t) = \left(\frac{b_i(t)}{d_i(t)}\right)_{i=1}^{d}$. If these two curves intersect for some parametric value $t_0$, then $\varphi_1$ and $\varphi_2$ satisfy $$\frac{a_i(t_0)}{c_i(t_0)} = \frac{b_i(t_0)}{d_i(t_0)}, \quad i=1, \ldots, d.$$
Hence finding the intersections amounts to computing the real roots of the univariate polynomials $f_i = a_i(t)d_i(t) - b_i(t)c_i(t)$,  which can be done with arbitrary precision by a wide range of numerical algorithms. The farthest-point Voronoi cells are unbounded convex polytopes, i.e., intersection of finitely many half-spaces. Once one has the list of the bounding hyperplanes of each farthest-point Voronoi cell in implicit form, finding intersections between a curve and the farthest-point Voronoi faces require computing intersections between the curve and certain bounding hyperplanes. Computing intersection between a rational parametric curve $\varphi(t) = \left(\frac{a_i(t)}{b_i(t)}\right)_{i=1}^{d}$ and a hyperplane $\sum_{i=1}^{d}c_ix_i = c$, is equivalent to solving the equation $$\sum_{i=1}^{d}\frac{a_i(t)c_i}{b_i(t)} = c,$$ which again reduces to the root (real) finding problem of a univariate polynomial.

Given the smallest enclosing ball of $S$, the intersections between $\nu(t)$ and $\partial(SEB(S))$ can be computed similarly. Then from the continuity of $\nu(t)$ the intersections can be easily characterized as ‘IN’-point or ‘OUT’-point by checking where $\nu(t)$ lies (inside or outside the $SEB(S)$) at any instant $t$ between two consecutive intersections. The pathological case where the curve $\nu(t)$ tangentially touches the boundary of $SEB(S)$ is to be discarded from $L$.

When the center moves along the arc $\xi(S', \nu(t))$, in step 3(a) it is required to check for intersection between $\xi(S'', \nu(t))$ and $\xi(S', \nu(t))$ for each nonempty $S'' \subset S'$. Since the static points are assumed to be in general position at most $d+1$ static points can appear on the boundary of the ball. So assuming that intersection between any two of such curves can be found in some constant time, computing all the intersections in step 3(a) requires $O(1)$ time (with respect to the input size $n$). 

Step 3(b) computes intersections of $\xi(S', \nu(t))$ and the farthest-point Voronoi faces involving $S'$. Suppose $|S'| = 1$, say $S' = \{s_1\}$. It is sufficient to compute intersection between $\xi(s_1, \nu(t))$ and the farthest-point Voronoi $(d-1)$-faces involving $s_1$ only, because all the lower dimensional faces associated with $s_1$ are incident to it. Clearly there could be at most $n - 1$ of such $(d-1)$-faces. So in the worst case, suppose that the farthest-point Voronoi cell of $s_1$ is given in the form of intersection of $n - 1$ half-spaces 
$$ c_ix_i^T \leq 0, i = 1, \ldots, n-1$$
where $c_i = (c_{i1}, c_{i2}, \ldots, c_{i(d+1)})$ and $x_i = (x_{i1}, x_{i2}, \ldots, x_{id}, 1)$. Each of the bounding hyperplanes $c_ix_i^T = 0$  contain one Voronoi $(d-1)$-face. To find the intersection between the curve and the $(d-1)$-face contained in the hyperplane $c_ix_i^T = 0$, we have to compute the intersection between the curve and the hyperplane $c_ix_i^T = 0$, and check if the solution satisfies the remaining $n - 2$ inequalities. Among the remaining $n - 2$ inequality relations, each time an equality is reported the dimension of the farthest-point Voronoi face is reduced by one. Note that each of these bounding hyperplanes is an orthogonal bisector hyperplane of $s_1$ and an input static point. So once these bounding hyperplanes along with their two defining input points is given, intersections of the curve with the farthest-point Voronoi faces involving $s_1$ is found by computing its intersections with the bounding hyperplanes and checking some inequalities or equalities. In the worst case it takes $O(n^2)$ time.   
Similarly, when $|S'| > 1$, we only need to do the same for any one $s_i \in S'$ only, since farthest-point Voronoi faces involving $S'$ is contained in the set of farthest-point Voronoi faces involving any $s_i \in S'$. Hence computing the intersections in step 3(b) requires $O(n^2)$ time in the worst case.

Also note that since we want the earliest (in $(t_{e_{last}}, T]$) of all these intersections, after finding intersections between two curves (or a curve and a farthest-point Voronoi) we only take earliest in $(t_{e_{last}}, T]$. As the number of such intersections is bounded (with the bound depending only upon the degree of the polynomials in the rational parametric function), it take constant time to find the earliest intersection. Thus step 3 requires $O(n)$ computations in the worst case.

For the `else' part in step 4, suppose that at some event point $t$, $S' \subset S$ lies on the boundary of $SEB(S \cup \{\nu(t)\})$. For sufficiently small $\epsilon > 0$, $S \setminus S'$ is contained in the interior of $SEB(S \cup \{\nu(t + \epsilon)\})$. In order to determine which of the points of $S'$ remain on the boundary of the ball after the $\epsilon$-perturbation, we need to find $S'' \subset S'$ such that $T = S'' \cup \{\nu(t + \epsilon)\}$ satisfies the necessary and sufficient conditions of lemma \ref{Alg1}, i.e., $S \subset CB(T)$ and $cc(T) \in conv(T)$. By the formula in lemma \ref{cc}, $cc(T)$ can be written as an affine combination of the points of $T$. $cc(T) \in conv(T)$ if and only if each coefficient of the affine combination is non-negative. 


For a successful execution of the algorithm the $\epsilon > 0$ is to be chosen so small that the perturbation doesn't skip the next event point. Also recall that the intersection computations were done numerically with arbitrary precision. If the error bound in calculating the intersection is $\delta$ then $\epsilon$ needs to be greater than $\delta$.

\begin{theorem}\label{Complex}
       Given as inputs the positions of the static points $S$ and the rational parametric function $\nu(t)$ defining the motion of the mobile point, the algorithm computes the parametric equation of the Euclidean 1-center function. 
       
\end{theorem}

\begin{proof}

Assume that $\mathcal{FPVD}(S)$, $SEB(S)$ and $SEB(S \cup \nu(0))$ are given or already computed. 

\emph{Initialization:} For the initialization of step 3, the required event point is set in step 2 and the first arc of the 1-center function is given. The rest is obvious.

\emph{Termination:} The program terminates if and only if it executes the END command in step 4 when no event point is found. The candidates for the event points are all obtained from the intersections computed in step 1 and step 3. Since $S$ is finite and $\nu(t)$ is a rational parametric function, there are only finitely many intersections in the interval $[0,T]$. Hence at some stage no event point will be found.

\emph{Maintenance:} It follows from the earlier discussions, that in each iterative step with the current arc of the 1-center and the last event point known, step 3 finds the next event point and step 4 determines the equation of the subsequent arc. $\Box$


\end{proof}

\subsection{Discussion for the case of multiple mobile points}\label{algo2}

The main governing principle being the same, the algorithm described in section \ref{algo1} can be modified for the case of multiple mobile points. However, there are certain stark differences between the case of a single mobile point and that of multiple mobile points, making the later particularly complicated and inefficient. Similar to the case of a single mobile point, in each iteration given the present arc of the Euclidean 1-center function, we have to consider each of the possible candidates for the next arc and compute their intersections. But the number of cases to consider in each step to find the next event point becomes quite large in the case of multiple mobile points, which makes the algorithm highly inefficient. We discuss in the following the differences between the single and the multiple mobile point cases and how the algorithm needs to be modified.

 Suppose that at some step the static points $S' \subseteq S$ is on the boundary of the $SEB$. Since the static points are assumed to be in general position, $|S'| \leq d +1$. However, any number of mobile points can appear on the boundary at any time. In that case the support set of the arc is a maximum subset of the set of points on the boundary of the $SEB$ that remain affinely independent throughout the corresponding time subinterval, which can be obtained in the usual greedy manner. Note that the choice of the support set , i.e., the maximum affinely independent set should not hamper the algorithm because if for two different support sets $P_1$ and $P_2$, the circumcenter functions $\xi(P_1)$ and $\xi(P_2)$ are equal in some interval then they are equal in the entire domain as they are rational parametric functions. As earlier, we assume that for initialization, the first arc of the 1-center function is known. Now suppose that at some step the arc of the 1-center function is given by $\xi(S', V')$, with $|S' \cup V'| \leq d +1$. In order to find the next event point we consider all the possible candidates for the next arc and compute their intersections. A change of arc of the 1-center function is caused by some combinatorial change of the points on the boundary of the $SEB$, i.e., some points appear on or leave the boundary of the $SEB$.

 No mobile point is on the boundary of the $SEB$ at some instant $t$ if and only if they are all in the interior of $SEB(S)$ at $t$. Similar to step 1 of the algorithm in section \ref{algo1}, we first have to compute the intersections of each $\nu_i(t)$ with $\partial(SEB(S))$ and identify the time intervals when all the mobile points are in the interior of $SEB(S)$. The left and right extremities of these time intervals will be the analogues of the `IN' and `OUT' points. 
 
 Note that since in the multiple mobile point case any number of points can appear on the boundary of the $SEB$, unlike the case of a single mobile point there are situations where the support sets of two consecutive arcs may be completely unrelated. Hence for step 3, to determine the subsequent arc we need to consider $\emph{all}$ subsets of $S \cup V$ of cardinality atleast 2 and no more than $d + 1$. This amounts to $O((n + m)^{d + 1})$ many cases to consider.

In case of a single mobile point, we averted heavy computations by computing intersections of the arc with the farthest-point Voronoi faces to detect addition of new static points on the boundary of the $SEB$. But this not possible in the present case. When $S' \neq \emptyset$, addition of new static points on the boundary of the $SEB$ can be detected separately by computing intersections between $\xi(S', V')$ and the farthest-point Voronoi faces involving $S'$. But when $S' = \emptyset$, addition of a single mobile point on the boundary of the $SEB$ does not cause an intersection with a farthest-point Voronoi face. Also as discussed in section \ref{theory}, when $S' = \emptyset$ intersection of the curve with a farthest-point Voronoi face does not always imply addition of static points. This situation never arose in the single mobile point case as the $SEB$ always had at least one static point on its boundary.

\section{Conclusion and future works}

In this paper we have proposed an algorithm that computes the parametric equation of the Euclidean 1-center function in $\mathbb{R}^d$, $d \geq 2$, for a system of $n$ static points  and $m$ mobile points having motion defined by rational parametric functions. In case of a single mobile point, given an arc of the 1-center function our algorithm computes the subsequent arc in $O(n^2)$ time (assuming that the farthest-point Voronoi diagram and the smallest enclosing ball of the set of static points is given or already computed). However, in case of multiple mobile points the algorithm loses its efficiency, especially for large values of $d$, due to the large number of exhaustive cases that are needed to consider in different iterations. The immediate course of future research would be to improve on the average complexity of the algorithm.


 Also there could be real life problems where it may not be appropriate to model the mobile and static agents as points. For example in the context of multi-robot systems it may be more befitting to consider the robots as disks in the plane. It is easy to see that the center of the smallest enclosing ball of a set of balls of same size and that of their centers coincide. But this does not hold true for a set of balls of different sizes. However, analogues of lemma \ref{WhereItLies1} and \ref{WhereItLies2} can be proved using the following observation: the center of the smallest enclosing ball of a set of balls lies on a face of the additively weighted Voronoi diagram of their centers weighted by the corresponding radii. Devising an efficient algorithm for this variant of the problem is another interesting direction for future research.


\bibliography{sebbib}

\end{document}

%% file: images/Fig1.eps_tex
\begingroup%
  \makeatletter%
  \providecommand\color[2][]{%
    \errmessage{(Inkscape) Color is used for the text in Inkscape, but the package 'color.sty' is not loaded}%
    \renewcommand\color[2][]{}%
  }%
  \providecommand\transparent[1]{%
    \errmessage{(Inkscape) Transparency is used (non-zero) for the text in Inkscape, but the package 'transparent.sty' is not loaded}%
    \renewcommand\transparent[1]{}%
  }%
  \providecommand\rotatebox[2]{#2}%
  \ifx\svgwidth\undefined%
    \setlength{\unitlength}{956.1341894bp}%
    \ifx\svgscale\undefined%
      \relax%
    \else%
      \setlength{\unitlength}{\unitlength * \real{\svgscale}}%
    \fi%
  \else%
    \setlength{\unitlength}{\svgwidth}%
  \fi%
  \global\let\svgwidth\undefined%
  \global\let\svgscale\undefined%
  \makeatother%
  \begin{picture}(1,0.61858772)%
    \put(0,0){\includegraphics[width=\unitlength]{Fig1.eps}}%
    \put(0.82652514,0.46289587){\color[rgb]{0,0,0}\makebox(0,0)[lt]{\begin{minipage}{0.07105399\unitlength}\raggedright \end{minipage}}}%
    \put(0.83516994,0.42165961){\color[rgb]{0,0,0}\makebox(0,0)[lb]{\smash{$B_1$}}}%
    \put(0.12329958,0.47856954){\color[rgb]{0,0,0}\makebox(0,0)[lt]{\begin{minipage}{0.11076064\unitlength}\raggedright $B_0$\end{minipage}}}%
    \put(0.31765312,0.13374867){\color[rgb]{0,0,0}\makebox(0,0)[lb]{\smash{$B_\lambda$}}}%
    \put(0.28839563,0.29884472){\color[rgb]{0,0,0}\makebox(0,0)[lt]{\begin{minipage}{0.08568276\unitlength}\raggedright \end{minipage}}}%
    \put(0.68128235,0.25704824){\color[rgb]{0,0,0}\makebox(0,0)[lb]{\smash{$c_1$}}}%
    \put(0.33228193,0.25077878){\color[rgb]{0,0,0}\makebox(0,0)[lb]{\smash{$c_0$}}}%
    \put(0.41378502,0.30511418){\color[rgb]{0,0,0}\makebox(0,0)[lb]{\smash{$c_\lambda$}}}%
    \put(0.55380318,0.46812042){\color[rgb]{0,0,0}\makebox(0,0)[lb]{\smash{z}}}%
  \end{picture}%
\endgroup%